\documentclass[submission,copyright,creativecommons]{eptcs}
\providecommand{\event}{GandALF 2023} % Name of the event you are submitting to
\pdfoutput=1 %uncomment to ensure pdflatex processing (mandatory e.g.\ to submit to arXiv)

\usepackage{iftex}

\ifpdf
  \usepackage{underscore}         % Only needed if you use pdflatex.
  \usepackage[T1]{fontenc}        % Recommended with pdflatex
\else
  \usepackage{breakurl}           % Not needed if you use pdflatex only.
\fi

\usepackage{todonotes}
\usepackage[utf8]{inputenc}
\usepackage[T1]{fontenc}
\usepackage{amssymb,amsmath,amsthm}
\usepackage{xspace}
\usepackage{booktabs}
\usepackage[capitalize,nameinlink]{cleveref}

\newtheorem{lemma}[]{Lemma}
\newtheorem{theorem}[]{Theorem}
\newtheorem{remark}[]{Remark}

\newtheorem{example}[]{Example}

\newtheorem{transformation}{Transformation}

\usepackage{tikz}

\usetikzlibrary{shadows.blur}
\usetikzlibrary{decorations.pathreplacing}
\usetikzlibrary{automata}
\tikzset{>=stealth, shorten >=1pt}
\tikzset{every edge/.style = {thick, ->, draw}}
\tikzset{every loop/.style = {thick, ->, draw}}

\tikzset{mystate/.style={draw,inner sep=3,circle}}

\newcommand{\nats}{\mathbb{N}}
\newcommand{\size}[1]{|#1|}

\renewcommand{\epsilon}{\varepsilon}
\renewcommand{\phi}{\varphi}

\newcommand{\set}[1]{\{#1\}}

\renewcommand{\complement}[1]{\overline{#1}}

\newcommand{\bigo}{\mathcal{O}}

\newcommand{\delaygame}[1]{\Gamma\!_{k}(#1)}
\newcommand{\delaygamep}[1]{\Gamma\!_{k'}(#1)}

\newcommand{\SigmaI}{\Sigma_I}
\newcommand{\SigmaO}{\Sigma_O}

\newcommand{\stratO}{\tau_O}
\newcommand{\stratI}{\tau_I}

\newcommand{\p}{P}

\newcommand{\arenagame}{\mathcal{G}}
\newcommand{\win}{\mathrm{Win}}

\newcommand{\trans}{\rightarrow}
\newcommand{\translabel}[1]{\xrightarrow{{#1}}}

\newcommand{\stratC}{\tau_c}
\newcommand{\stratE}{\tau_e}

\newcommand{\play}{\mathrm{play}}

\newcommand{\pref}[2]{#1[#2]}
\newcommand{\prefs}{\mathrm{Pref}(\arenagame)}
\newcommand{\prefsE}{\mathrm{Pref}_e(\arenagame)}
\newcommand{\prefsC}{\mathrm{Pref}_c(\arenagame)}

\newcommand{\prob}[1]{{\cal P}\left(#1\right)}

\newcommand{\aut}{\mathcal{A}}
\newcommand{\init}{I}
\newcommand{\col}{\Omega}

\newcommand{\pspace}{{\upshape{\textsc{PSpace}}}\xspace}
\newcommand{\exptime}{{\upshape{\textsc{ExpTime}}}\xspace}

\newcommand{\twoexp}{{\upshape{\textsc{2ExpTime}}}\xspace}
\newcommand{\threeexp}{{\upshape{\textsc{3ExpTime}}}\xspace}

\newcommand\abbrv[1]{}

\title{Strategies Resilient to Delay:\\Games under Delayed Control vs.\ Delay Games}
\def\titlerunning{Strategies Resilient to Delay}

\author{Martin Fränzle
\institute{Carl v.\ Ossietzky Universität\\ Oldenburg, Germany}
\email{martin.fraenzle@uol.de}
\and
Sarah Winter
\institute{Université libre de Bruxelles\\ Brussels, Belgium}
\email{swinter@ulb.ac.be}
\and
Martin Zimmermann
\institute{Aalborg University\\ Aalborg, Denmark}
\email{mzi@cs.aau.dk}
}

\def\authorrunning{M.\ Fränzle, S.\ Winter \& M.\ Zimmermann}

\begin{document}

\maketitle
\begin{abstract}
We compare games under delayed control and delay games, two types of infinite games modelling asynchronicity in reactive synthesis. Our main result, the interreducibility of the existence of sure winning strategies for the protagonist, allows to transfer known complexity results and bounds on the delay from delay games to games under delayed control, for which no such results had been known.  We furthermore analyze existence of randomized strategies that win almost surely, where this correspondence between the two types of games breaks down.
\end{abstract}

%%%%%%%%%%%%%%%%%%%%%%%%%%%%%%%%%%%%%%%%%%%%%%%%%%%%%%%%%%%%%%%%%%%%%%%%%%%%%%%%%%%%%%%%%%%%
%%%%%%%%%%%%%%%%%%%%%%%%%%%%%%%%%%%%%%%%%%%%%%%%%%%%%%%%%%%%%%%%%%%%%%%%%%%%%%%%%%%%%%%%%%%%
%%%%%%%%%%%%%%%%%%%%%%%%%%%%%%%%%%%%%%%%%%%%%%%%%%%%%%%%%%%%%%%%%%%%%%%%%%%%%%%%%%%%%%%%%%%%
\section{Introduction}

Two-player zero-sum games of infinite duration are a standard model for the synthesis of reactive controllers, i.e., correct-by-construction controllers that satisfy their specification even in the presence of a malicious environment. 
In such games, the interaction between the controller and the environment is captured by the rules of the game and the specification on the controller induces the winning condition of the game. 
Then, computing a correct controller boils down to computing a winning strategy. 

Often, it is convenient to express the rules in terms of a graph capturing the state-space such that moves correspond to transitions between these states.
The interaction between the controller and the environment then corresponds to a path through the graph and the winning condition is a language of such paths, containing those that correspond to interactions that satisfy the specification on the controller.

In other settings, it is more convenient to consider a slightly more abstract setting without game graphs, so-called Gale-Stewart games~\cite{GS}.
In such games, the players alternatingly pick a sequence of letters, thereby constructing an infinite word.
The winning condition is a language over infinite words, containing the winning words for one player. 
To capture the synthesis problem, the winning condition has to encode both the specification on the controller as well as the rules of interaction. 
It is straightforward to transform a graph-based game into a Gale-Stewart game and a Gale-Stewart game into a graph-based game such that the existence of winning strategies for both players is preserved.

In the most basic setting of synthesis, both the controller and the environment are fully informed about the current state of the game (complete information). 
However, this scenario is not always realistic. 
Thus, much effort has been poured into studying games under incomplete information where the players are only partially informed about the current state of the game.
Here, we are concerned with a special type of partial information designed to capture delays in perception and action. Such delays either render the most recent moves of the opponent invisible to a player or induce a time lag between the selection and the implementation of an own move, respectively.

As a motivating example, consider the domain of cooperative driving: Here, the exchange of information between cars is limited (and therefore delayed) by communication protocols that have to manage the available bandwidth to transfer information between cars. Other delaying factors include, e.g., complex signal processing chains based on computer vision to detect the locations of obstacles.
Thus, decisions have to be made based on incomplete information, which only arrives after some delay.

\textbf{Games under Delayed Control.}
Chen et al.~\cite{ChenFLMZ21} introduced (graph) games under delayed control to capture this type of incomplete information. 
Intuitively, assume the players so far have constructed a finite path~$v_0 \cdots v_k$ through the graph. 
Then, the controller has to base her decision on a visible proper prefix $v_0 \cdots v_{k-\delta}$, where $\delta$ is the amount of delay. 
Hence, the suffix~$v_{k-\delta+1} \cdots v_k$ is not yet available to base the decision on, although the decision to be made is to be applied at the last state $v_k$ in the sequence.

They showed that solving games under delayed control with safety conditions and with respect to  a given delay is decidable: 
They presented two algorithms, an exponential one based on a reduction to delay-free safety games using a queue of length~$\delta$, and a more practical incremental algorithm synthesizing a series of
controllers handling increasing delays and reducing game-graph size in between. 
They showed that even a naïve implementation of this algorithm outperforms the reduction-based one, even when the latter is used with state-of-the-art solvers for delay-free games.
However, the exact complexity of the incremental algorithm and that of solving games under delayed control remained open.

Note that asking whether there is some delay~$\delta$ that allows controller to win reduces to solving standard, i.e., delay-free games, as they correspond to the case~$\delta = 0$. The reason is monotonicity in the delay: if the controller can win for delay $\delta$ then also for any $\delta' < \delta$.
More interesting is the question whether controller wins with respect to  every possible delay. 
Chen et al.\ conjectured that there is some exponential $\delta$ such that if the controller wins under delay~$\delta$, then also under every $\delta'$.

\textbf{Delay Games.}
There is also a variant of Gale-Stewart games modelling delayed interaction between the players \cite{DBLP:conf/icalp/HoschL72}. 
Here, the player representing the environment (often called Player~$I$) has to provide a lookahead on her moves, i.e., the player representing the controller (accordingly called Player~$O$) has access to the first $n+k$ letters picked by Player~$I$ when picking her $n$-th letter.
So, $k$ is the amount of lookahead that Player~$I$ has to grant Player~$O$. 
Note that the lookahead benefits Player~$O$ (representing the controller) while the delay in a game under delayed control disadvantages the controller.

Only three years after the seminal Büchi-Landweber theorem showing that delay-free games with $\omega$-regular winning conditions are decidable~\cite{BL69}, Hosch and Landweber showed that it is decidable whether there is a $k$ such that Player~$O$ wins a given Gale-Stewart game with lookahead~$k$~\cite{DBLP:conf/icalp/HoschL72}. 
Forty years later, Holtmann, Kaiser, and Thomas~\cite{HoltmannKT12} revisited these games (and dubbed them delay games). 
They  proved that if Player~$O$ wins a delay game then she wins it already with at most doubly-exponential lookahead (in the size of a given deterministic parity automaton recognizing the winning condition).
Thus, unbounded lookahead does not offer any advantage over doubly-exponential lookahead in games with $\omega$-regular winning conditions.
Furthermore, they presented an algorithm with doubly-exponential running time solving delay games with $\omega$-regular winning conditions conditions, i.e., determining whether there exists a $k$ such that Player~$O$ wins a given delay game (with its winning condition again given by a deterministic parity automaton) with lookahead~$k$.

Both upper bounds were improved and matching lower bounds were later proven by Klein and Zimmermann~\cite{KleinZ14}: Solving delay games is \exptime-complete and exponential lookahead is both necessary to win some games and sufficient to win all games that can be won.
Both lower bounds already hold for winning conditions specified by deterministic safety automata while the upper bounds hold for deterministic parity automata.
The special case of solving games with conditions given as reachability automata is \pspace-complete, but exponential lookahead is still necessary and sufficient.
Thus, there are tight complexity results for delay games, unlike for games under delayed control.

\textbf{Our Contributions.}
In this work, we exhibit a tight relation between controller in a game under delayed control and Player~$I$ in a delay game (recall that these are the players that are disadvantaged by delay and lookahead, respectively).
Note that winning conditions in games under delayed control are always given from the perspective of controller (i.e., she has to avoid unsafe states in a safety game) while winning conditions in delay games are always given from the perspective of Player~$O$. 
Hence, as we relate controller and Player~$I$, we always have to complement winning conditions.

More precisely, we show that one can transform a safety game under delayed control in polynomial time into a delay game with a reachability condition for Player~$O$ (i.e., with a safety condition for Player~$I$) such that controller wins the game under delayed control with delay~$\delta$ if and only if Player~$I$ wins the resulting delay game with lookahead of size~$\frac{\delta}{2}$. 
Dually, we show that one can transform a delay game with safety condition for Player~$I$ in polynomial time into a reachability game under delayed control such that Player~$I$ wins the delay game with lookahead of size~$\delta$ if and only if controller wins the resulting game under delayed control with delay~$2\delta$. 
Thus, we can transfer both upper and lower bound results on complexity and on (necessary and sufficient) lookahead from delay games to delayed control. 
In particular, determining whether controller wins a given safety game under delayed control for every possible delay is \pspace-complete. 
Our reductions also prove the conjecture by Chen et al.\ on the delays that allow controller to win such games.
Furthermore, we generalize our translation from games with safety conditions to games with parity conditions and games with winning conditions given by formulas of Linear Temporal Logic (LTL)~\cite{Pnueli}, again allowing us to transfer known results for delay games to games under delayed control.

Note that we have only claimed that the existence of winning strategies for the controller in the game under delayed control and Player~$I$ in the delay game coincides. 
This is no accident! 
In fact, the analogous result for relating environment and Player~$O$ fails. 
This follows immediately from the fact that delay games are determined while games under delayed control are undetermined, even with safety conditions. The reason is that the latter games are truly incomplete information games (which are typically undetermined) while delay games are perfect information games. 

We conclude by a detailed comparison between environment and Player~$O$ in both the setting with deterministic as well as in the setting with randomized strategies. 
The latter setting increases power for both the controller and the environment, making them win (almost surely) games under delayed control that remain undetermined in the deterministic setting, but it also breaks the correspondence between controller and Player~$I$ observed in the deterministic setting: there are games that controller wins almost surely while Player~$I$ surely looses them.

All proofs which are omitted due to space restrictions can be found in the full version~\cite{arxiv}.

%%%%%%%%%%%%%%%%%%%%%%%%%%%%%%%%%%%%%%%%%%%%%%%%%%%%%%%%%%%%%%%%%%%%%%%%%%%%%%%%%%%%%%%%%%%%
%%%%%%%%%%%%%%%%%%%%%%%%%%%%%%%%%%%%%%%%%%%%%%%%%%%%%%%%%%%%%%%%%%%%%%%%%%%%%%%%%%%%%%%%%%%%
%%%%%%%%%%%%%%%%%%%%%%%%%%%%%%%%%%%%%%%%%%%%%%%%%%%%%%%%%%%%%%%%%%%%%%%%%%%%%%%%%%%%%%%%%%%%
\section{Preliminaries}
\label{sec:prels}

We denote the non-negative integers by~$\nats$.
An {alphabet}~$\Sigma$ is a non-empty finite set of {letters}.
A {word} over $\Sigma$ is a finite or infinite sequence of letters of $\Sigma$: 
The set of finite words (non-empty finite words, infinite words) over $\Sigma$ is denoted by $\Sigma^*$ ($\Sigma^+$, $\Sigma^\omega$).
The {empty word} is denoted by $\varepsilon$,
the length of a finite word~$w$ is denoted by~$\size{w}$.
Given two infinite words~$\alpha \in (\Sigma_0)^\omega$ and $\beta \in (\Sigma_1)^\omega$, we define $\binom{\alpha}{\beta} = \binom{\alpha(0)}{\beta(0)}\binom{\alpha(1)}{\beta(1)}\binom{\alpha(2)}{\beta(2)} \cdots \in (\Sigma_0 \times \Sigma_1)^\omega$. 

%%%%%%%%%%%%%%%%%%%%%%%%%%%%%%%%%%%%%%%%%%%%%%%%%%%%%%%%%%%%%%%%%%%%%%%%%%%%%%%%%%%%%%%%%%%%
%%%%%%%%%%%%%%%%%%%%%%%%%%%%%%%%%%%%%%%%%%%%%%%%%%%%%%%%%%%%%%%%%%%%%%%%%%%%%%%%%%%%%%%%%%%%
%%%%%%%%%%%%%%%%%%%%%%%%%%%%%%%%%%%%%%%%%%%%%%%%%%%%%%%%%%%%%%%%%%%%%%%%%%%%%%%%%%%%%%%%%%%%
\subsection{Games under Delayed Control}
\label{subsec:games}

Games under delayed control are played between two players, controller and environment. 
For pronomial convenience~\cite{McNaughton00}, we refer to controller as she and environment as he.

A game~$\arenagame = ( S, s_0, S_c, S_e, \Sigma_c, \Sigma_e, \trans, \win)$ consists of a finite set~$S$ of states partitioned into the states~$S_c \subseteq S$ of the controller and the states~$S_e \subseteq S$ of the environment, an initial state~$s_0 \in S_c$, the sets of actions~$\Sigma_c$ for the controller and $\Sigma_e$ for the environment, a transition function~$\trans\colon (S_c \times \Sigma_c) \cup (S_e \times \Sigma_e) \rightarrow S$ such that $s \in S_c$ and $\sigma \in \Sigma_c$ implies $\trans\!\!(s,\sigma) \in S_e$ and vice versa, and a winning condition~$\win \subseteq S^\omega$.
We write $s \translabel{\sigma} s'$ as shorthand for $s' =\, \trans\!\!(s, \sigma)$.

A play in $\arenagame$ is an infinite sequence~$\pi = \pi_0 \sigma_0  \pi_1 \sigma_1  \pi_2 \sigma_2 \cdots $ satisfying $\pi_0 = s_0$ and $\pi_n \translabel{\sigma_n} \pi_{n+1}$ for all $n \ge 0$.
We say that controller wins $\pi$ if $\pi_0 \pi_1 \pi_2 \cdots \in \win$; otherwise, we say that environment wins $\pi$.
The play prefix of $\pi$ of length~$n$ is defined as $\pref{\pi}{n} = \pi_0 \sigma_0 \cdots \sigma_{n-1} \pi_n$, i.e., $n$ is the number of actions (equivalently, the number of transitions).
We denote by $\prefs$ the set of play prefixes of all plays in $\arenagame$, which is partitioned into the sets~$\prefsC$ and $\prefsE$ of play prefixes ending in $S_c$ and $S_e$, respectively. 
Due to our alternation assumption, play prefixes of even (odd) length are in $\prefsC$ ($\prefsE$).

Fix some even $\delta \ge 0$. A strategy for the controller in $\arenagame$ under delay~$\delta$ is a pair $(\alpha, \stratC)$ where $\alpha \in (\Sigma_c)^{\frac{\delta}{2}}$ and $\stratC \colon \prefsC \rightarrow \Sigma_c$ maps play prefixes ending in $S_c$ to actions of the controller.
A play~$\pi_0 \sigma_0 \pi_1 \sigma_1 \pi_2 \sigma_2 \cdots$ is consistent with $(\alpha,\stratC)$ if $\sigma_0 \sigma_2 \cdots \sigma_{\delta -4} \sigma_{\delta-2} = \alpha$ and $\sigma_{2n} = \stratC(\pref{\pi}{2n - \delta})$ for all $2n > \delta -2$, i.e., controller has access to environment's actions with a delay of $\delta$.
In particular, her first $\frac{\delta}{2} +1$ actions are independent of environment's actions and, in general, her $n$-th action~$\sigma_{2n}$ only depends on the actions~$\sigma_1, \ldots, \sigma_{(2n-\delta) -1}$ picked by environment, but not on the actions~$\sigma_{(2n-\delta)+1},\ldots, \sigma_{2n-1}$.
The strategy~$(\alpha, \stratC)$ is winning under delay~$\delta$ if every play that is consistent with it is winning for controller.
Controller wins $\arenagame$ under delay~$\delta$ if she has a winning strategy under delay~$\delta$ for $\arenagame$.

\begin{remark}
\label{remark:delaycontrol} \hfill
 \begin{enumerate}
    \item 
The notion of winning strategy for controller under delay~$0$ is the classical one for delay-free games (cf.~\cite{GTW02}).
    \item
\label{remarkitem:delaycontrol:monotonicity}  If 
controller wins $\arenagame$ under delay~$\delta$, then also under every delay~$\delta' < \delta$~\cite{ChenFLMZ21}.
 \end{enumerate}
\end{remark}

A strategy for environment is a mapping~$\stratE\colon \prefsE \rightarrow \Sigma_e$.
A play~$\pi_0 \sigma_0 \pi_1 \sigma_1 \pi_2 \sigma_2 \cdots$ is consistent with $\stratE$ if $\sigma_{2n+1} = \stratE(\pi_0 \sigma_0 \cdots \sigma_{2n-1} \pi_{2n+1})$ for all $n \ge 0$, i.e., environment has access to the full play prefix when picking his next action.
The strategy~$\stratE$ is winning, if every play that is consistent with it is winning for the environment (i.e., the sequence of states is not in $\win$).
Further, we say that environment wins $\arenagame$, if he has a winning strategy for $\arenagame$. 
Note that the two definitions of strategies are in general not dual, e.g., the one for environment is not defined with respect to a delay~$\delta$. 
% In fact, the notion of winning strategy for environment is the classical one for delay-free games (cf.~\cite{GTW02}).
 \begin{remark}
 The notion of winning strategy for environment is the classical one for delay-free games (cf.~\cite{GTW02}).
 \end{remark}

We say that a game under delayed control~$\arenagame$ is determined under delay~$\delta$, if either controller wins $\arenagame$ under delay~$\delta$ or environment wins $\arenagame$.
Let us stress that determinacy is defined with respect to some fixed~$\delta$ and that $\arenagame$ may be determined for some $\delta$, but undetermined for some other~$\delta'$ (due to the non-dual definition of strategies).
\cref{remark:undetermined} shows an undetermined safety (!) game under delayed control.

\begin{example}
\label{example:gameunderdelayedcontrol}
Consider the game~$\arenagame = (S, s_\init, S_c, S_e, \Sigma_c, \Sigma_e, \rightarrow, \win)$ depicted in \cref{fig:gameunderdelayedcontrol} where $\win $ contains all plays that do not visit the black vertex.
Note that this is a safety condition. 
In particular, if controller does not pick action~$b$ at $c_2$ and does not pick action~$a$ at $c_3$, then the vertex $e_3$ is never reached.
This is straightforward without delay, but we claim that controller can also win $\arenagame$ under delay~$2$.

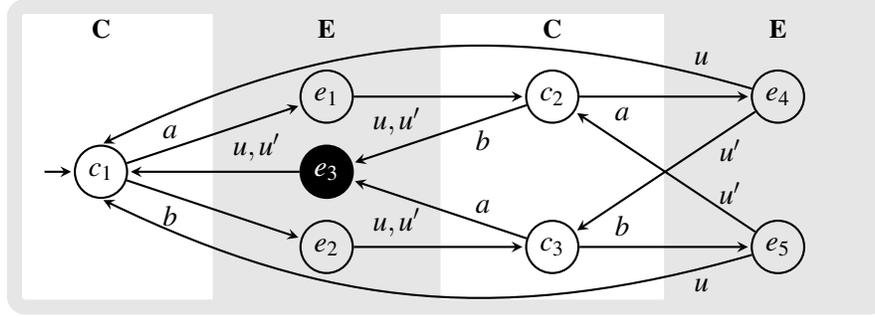
\begin{figure}
    \centering
\begin{tikzpicture}[thick]
    
    \draw[gray!20, rounded corners,line width=2mm] (-1.15,-1.8) rectangle (10.35,2.2);  
    
    \draw[fill=gray!20,gray!20] (1.5,-1.8) rectangle (4.5,2.2); 
    \draw[fill=gray!20,gray!20] (7.5,-1.8) rectangle (10.3,2.2); 
    
    \node[black] at (0,1.9) {\small \bf C};
    \node[black] at (3,1.9) {\small \bf E};
    \node[black] at (6,1.9) {\small \bf C};
    \node[black] at (9,1.9) {\small \bf E};

    \node[mystate] (c1) at (0,0) {$c_1$};
    \node[mystate] (e1) at (3,1) {$e_1$};
    \node[mystate] (e2) at (3,-1) {$e_2$};
    \node[mystate] (c2) at (6,1) {$c_2$};
    \node[mystate] (c3) at (6, -1) {$c_3$};
    \node[mystate] (e4) at (9,1) {$e_4$};
    \node[mystate,fill=black,text=white] (e3) at (3,0) {$e_3$};
    \node[mystate] (e5) at (9,-1) {$e_5$};

     \path[-stealth]
     (-.75,0) edge (c1)
     (c1) edge[near start] node[above] {$a$} (e1)
     (c1) edge[near start] node[below] {$b$} (e2)
     (e1) edge[near start] node[below] {$u,u'$} (c2)
     (e2) edge[near start] node[above] {$u,u'$} (c3)
     (c2) edge[near start] node[below] {$a$} (e4)
     (c2) edge[near start] node[below] {$b$} (e3)
     (c3) edge[near start] node[above] {$a$} (e3)
     (c3) edge[near start] node[above] {$b$} (e5)
     (e3) edge[] node[above,near start] {$u,u'$} (c1)
     (e4) edge[bend right =21] node[above,pos=.07] {$u$} (c1.north)
     (e4) edge[bend right=0] node[very near start,below] {$u'$} (c3)
     (e5) edge[bend left =21] node[above,below,pos=.07] {$u$} (c1.south)
     (e5) edge[bend left=0] node[above,very near start] {$u'$} (c2)
    ;

\end{tikzpicture}
    \caption{The game for \cref{example:gameunderdelayedcontrol}. Controller wins all plays that never visit the black vertex. Note that we have $\Sigma_c =\set{a,b} $ and $\Sigma_e = \set{u,u'}$.}
    \label{fig:gameunderdelayedcontrol}
\end{figure}

To gain some intuition, consider a play prefix~$\pi_0 \sigma_0 \pi_1 \cdots \pi_{n-1}\sigma_{n-1}\pi_{n}$ with $n\ge 4$ and $\pi_n \in S_c$.
Then, controller has to pick an action~$\sigma_n$ to continue the prefix.
However, due to the delayed control, she has to do so based on the prefix~$\pi_0 \sigma_0 \pi_1 \cdots \pi_{n-3}\sigma_{n-3}\pi_{n-2}$.

If $\pi_{n-2}$ is $c_2$, then $\pi_n$ is either $c_3$ or $c_1$. 
Hence, picking $\sigma_n = b$ is the only safe choice.
Dually, if $\pi_{n-2}$ is $c_3$, then $\pi_n$ is either $c_2$ or $c_1$. 
Hence, picking $\sigma_n = a$ is the only safe choice.

Finally, assume $\pi_{n-2}$ is $c_1$.
Then, $\pi_n$ is either $c_2$ or $c_3$. 
In the former case, picking $\sigma_n = a$ is the only safe choice, in the latter case, picking $\sigma_n = b$ is the only safe choice. 
So, controller needs to distinguish these two cases, although she has no access to $\pi_n$.

But she can do so by inspecting~$\pi_{n-3}$ (which she has access to): 
As a predecessor of $\pi_{n-2} = c_1$, it can either be $e_4$, $e_5$, or $e_3$.
In the latter case, the play is already losing.
Thus, we disregard this case, as we construct a winning strategy. 
So, assume we have $\pi_{n-3} = e_4$ (the case~$\pi_{n-3} = e_5$ is dual).
Then, we must have $\pi_{n-4} = c_2$ (the only predecessor of $e_4$) and, by our analysis of the safe moves above, controller must have picked $\sigma_{n-2} = b$ (based, due to delay, on the prefix ending in $\sigma_{n-4} = c_2$). 
From this we can conclude $\pi_{n-1} = e_2$ and thus $\pi_n = c_3$ (the only successor of $e_2$).
Thus, she can safely pick $\sigma_{n} = b$.

This intuition, and the necessary initialization, is implemented by the strategy~$(\alpha,\stratC)$ with $\alpha = a$ and 
\[
\stratC(\pi_0 \sigma_0 \pi_1 \cdots \pi_{n-3}\sigma_{n-3}\pi_{n-2}) = \begin{cases}
a &\text{$n = 2$ and $\pi_0 = c_1$,} \\
b&\text{$n > 2$, $\pi_{n-2} = c_1$, and $\pi_{n-3} = e_4$,} \\
a&\text{$n > 2$, $\pi_{n-2} = c_1$, and $\pi_{n-3} = e_5$,} \\
b&\text{$\pi_{n-2} = c_2$,} \\
a&\text{$\pi_{n-2} = c_3$}.
\end{cases}
\]
An induction over the play length shows that $(\alpha,\stratC)$ is winning for controller under delay~$2$.
\end{example}

\begin{remark}
Our definition of games under delayed control differs in three aspects from the original definition of Chen et al.~\cite{ChenFLMZ21}.
\begin{itemize}

    \item We allow arbitrary winning conditions while Chen et al.\ focused on safety conditions.
    
    \item The original definition allows nondeterministic strategies (a strategy that returns a nonempty set of actions, each one of which can be taken), while we restrict ourselves here to deterministic strategies (a strategy that returns a single action to be taken). 
    The motivation for their use of nondeterministic strategies is the fact that they can be refined if additional constraints are imposed, which Chen et al.'s algorithm computing a winning strategy relies on.
    
    Here, on the other hand, we are just interested in the existence of winning strategies.
    In this context, it is sufficient to consider deterministic strategies, as controller has a nondeterministic winning strategy if and only if she has a deterministic winning strategy.
    Also, strategies in delay games are deterministic, so the transformation between games under delayed control and delay games can be formulated more naturally for deterministic strategies.
    
    \item The original definition also allowed odd delays~$\delta$ while we only allow even delays. 
    As we will see in \cref{sec:transformation}, the transformation of games under delayed control to delay games is naturally formulated for even delays.
    This choice also simplifies definitions, as accounting for odd delays imposes an additional notational burden.
    
\end{itemize}
\end{remark}

%%%%%%%%%%%%%%%%%%%%%%%%%%%%%%%%%%%%%%%%%%%%%%%%%%%%%%%%%%%%%%%%%%%%%%%%%%%%%%%%%%%%%%%%%%%%
%%%%%%%%%%%%%%%%%%%%%%%%%%%%%%%%%%%%%%%%%%%%%%%%%%%%%%%%%%%%%%%%%%%%%%%%%%%%%%%%%%%%%%%%%%%%
%%%%%%%%%%%%%%%%%%%%%%%%%%%%%%%%%%%%%%%%%%%%%%%%%%%%%%%%%%%%%%%%%%%%%%%%%%%%%%%%%%%%%%%%%%%%
\subsection{Delay Games}
\label{subsec:delaygames}

{Delay games} are played between two players, Player~$I$ (she) and Player~$O$ (he). 
A delay game~$\delaygame{L}$ (with constant lookahead) consists of a lookahead~$k \in \nats$ and a {winning condition}~$L \subseteq (\SigmaI \times \SigmaO)^\omega$ for some alphabets~$\SigmaI$ and $\SigmaO$. 
Such a game is played in rounds~$n = 0,1,2, \ldots$ as follows: in round~$0$, first Player~$I$ picks a word~$x_0 \in \SigmaI^{k+1}$, then Player~$O$ picks a letter~$y_0 \in \SigmaO$. 
In round~$n>0$, Player~$I$ picks a letter~$x_n \in \SigmaI$, then Player~$O$ picks a letter~$y_n \in \SigmaO$.
Player~$O$ {wins} a play~$(x_0, y_0)(x_1, y_1)(x_2, y_2) \cdots $ if the outcome~$\binom{ x_0 x_1 x_2 \cdots }{ y_0 y_1 y_2 \cdots }$ is in $L$; otherwise, Player~$I$ wins.

A {strategy} for Player~$I$ in $\delaygame{L}$ is a mapping $\stratI \colon \SigmaO^* \rightarrow \SigmaI^*$ satisfying $\size{\stratI(\epsilon)} = k+1$ and $\size{\stratI(w)} = 1$ for all $w \in \SigmaO^+$. 
A strategy for Player~$O$ is a mapping~$\stratO \colon \SigmaI^+ \rightarrow \SigmaO$. A play~$(x_0, y_0)(x_1, y_1)(x_2, y_2) \cdots $ is {consistent} with $\stratI$ if $x_n = \stratI(y_0 \cdots y_{n-1})$ for all $n \ge 0$, and it is consistent with $\stratO$ if $y_n = \stratO(x_0 \cdots x_n)$ for all $n \ge 0$. 
So, strategies are dual in delay games, i.e., Player~$I$ has to grant some lookahead on her moves that Player~$O$ has access to.
A strategy for Player~$\p \in \set{I,O}$ is {winning}, if every play that is consistent with the strategy is won by Player~$\p$.
We say that Player~$\p\in \set{I,O}$ {wins} a game  $\delaygame{L}$ if Player~$\p$ has a winning strategy in $\delaygame{L}$.

\begin{remark}
\label{remark:delaymonotonicity}
 \hfill
 \begin{itemize}
     \item 
    If Player~$O$ wins $\delaygame{L}$, then he also wins $\delaygamep{L}$ for every $k' > k$.
     \item 
    If Player~$I$ wins $\delaygame{L}$, then she also wins $\delaygamep{L}$ for every $k' < k$.
 \end{itemize}
\end{remark}

Unlike games under delayed control, delay games with Borel winning conditions are determined~\cite{KleinZ14}, i.e., each delay game~$\delaygame{L}$ with Borel~$L$ and fixed~$k$ is won by one of the players.

\begin{example}
Consider $
L=\left\{ \binom{a_0}{b_0}\binom{a_1}{b_1}\binom{a_2}{b_2}\cdots \mid b_0 \notin\set{a_0, a_1, a_2}  \right\}$ over the alphabets~$\SigmaI = \SigmaO = \set{1,2,3,4}$.

Player~$I$ wins $\delaygame{L}$ for $k=1$ with the following strategy~$\stratI$: $\stratI(\epsilon) = 12$ and $\stratI(b_0) = b_0$, and $\stratI(w)$ arbitrary for all $w \in \SigmaO^+$ with $\size{w} > 1$: In round~$0$, after Player~$I$ has picked~$a_0a_1 = 12$, Player~$O$ has to pick some $b_0$. In order to not loose immediately, he has to pick $b_0 \notin \set{1,2}$. Then, in round~$1$, Player~$I$ picks $a_2 = b_0$ and thereby ensures $b_0 \in \set{a_0, a_1, a_2}$. 
Hence, the play is not won by Player~$O$ (it's outcome is not in $L$), therefore it is winning for Player~$I$.

However, Player~$O$ wins $\delaygame{L}$ for $k = 2$ with the following strategy~$\stratO$: $\stratO(a_0a_1a_2)$ is a letter in the nonempty set~$\SigmaO \setminus \set{a_0, a_1, a_2}$ and $\stratO(w)$ arbitrary for all $w \in \SigmaI^*$ with $\size{w} \neq 3$.
In round~$0$, after Player~$I$ has picked~$a_0a_1a_2$, Player~$O$ picks $b_0 \notin \set{a_0, a_1, a_2}$ and thus ensures that the outcome is in $L$.
\end{example}

\begin{remark}
We restrict ourselves here to the setting of constant lookahead, i.e., in a delay game~$\delaygame{L}$ in round~$n$ when Player~$O$ picks her $n$-th letter, Player~$I$ has already picked $k+n+1$ letters (note that we start in round~$0$ with the zeroth letter).
Delay games have also been studied with respect to growing lookahead, i.e., the lookahead increases during a play~\cite{HoltmannKT12}.
However, it is known that constant lookahead is sufficient for all $\omega$-regular winning conditions: if Player~$O$ wins for any lookahead (no matter how fast it is growing), then she also wins with respect to constant lookahead, which can even be bounded exponentially in the size of a deterministic parity automaton recognizing the winning condition~\cite{KleinZ14}.
Stated differently, growing lookahead does not allow to win any more games than constant lookahead.
Finally, the setting of constant lookahead in delay games considered here is the natural counterpart to games under delayed control, where the delay is fixed during a play. 
\end{remark}

%%%%%%%%%%%%%%%%%%%%%%%%%%%%%%%%%%%%%%%%%%%%%%%%%%%%%%%%%%%%%%%%%%%%%%%%%%%%%%%%%%%%%%%%%%%%
%%%%%%%%%%%%%%%%%%%%%%%%%%%%%%%%%%%%%%%%%%%%%%%%%%%%%%%%%%%%%%%%%%%%%%%%%%%%%%%%%%%%%%%%%%%%
%%%%%%%%%%%%%%%%%%%%%%%%%%%%%%%%%%%%%%%%%%%%%%%%%%%%%%%%%%%%%%%%%%%%%%%%%%%%%%%%%%%%%%%%%%%%
\subsection{\texorpdfstring{$\mathbf{\omega}$-Automata}{omega-Automata}}
\label{subsec:automata}

A deterministic reachability automaton~$\aut = (Q, \Sigma, q_\init, \delta_\aut, F)$ consists of a finite set~$Q$ of states containing the initial state~$q_\init \in Q$ and the set of accepting states~$F \subseteq Q$, an alphabet~$\Sigma$, and a transition function~$\delta_\aut \colon Q \times \Sigma \rightarrow Q$.
The size of $\aut$ is defined as $\size{\aut} = \size{Q}$.
Let $w = w_0 w_1 w_2 \cdots \in \Sigma^\omega$.
The run of $\aut$ on $w$ is the sequence~$q_0 q_1 q_2 \cdots $ such that $q_0 = q_\init$ and $q_{n+1} = \delta_\aut(q_n, w_n)$ for all $n\ge 0$.
A run~$q_0q_1 q_2\cdots$ is (reachability) accepting if $q_n \in F$ for some $n \ge 0$.
The language (reachability) recognized by $\aut$, denoted by $L(\aut)$, is the set of infinite words over $\Sigma$ such that the run of $\aut$ on $w$ is (reachability) accepting.

A deterministic safety automaton has the form~$\aut = (Q, \Sigma, q_\init, \delta_\aut, U)$ where $Q, \Sigma, q_\init, \delta_\aut$ are as in a deterministic reachability automaton and where $U \subseteq Q$ is a set of unsafe states.
The notions of size and runs are defined as for reachability automata, too. 
A run~$q_0q_1q_2 \cdots$ is (safety) accepting if $q_n \notin U$ for all $n \ge 0$.
The language (safety) recognized by $\aut$, again denoted by $L(\aut)$, is the set of infinite words over $\Sigma$ such that the run of $\aut$ on $w$ is (safety) accepting.

A deterministic parity automaton has the form~$\aut = (Q, \Sigma, q_\init, \delta_\aut, \col)$ where $Q, \Sigma, q_\init, \delta_\aut$ are as in a deterministic reachability automaton and where $\col \colon Q \rightarrow \nats$ is a coloring of the states.
The notions of size and runs are defined as for reachability automata, too. 
A run~$q_0q_1q_2 \cdots$ is (parity) accepting if the maximal color appearing infinitely often in the sequence~$\col(q_0)\col(q_1)\col(q_2)\cdots $ is even.
The language (parity) recognized by $\aut$, again denoted by $L(\aut)$, is the set of infinite words over $\Sigma$ such that the run of $\aut$ on $w$ is (parity) accepting.

Reachability and safety automata are dual while parity automata are self-dual.

\begin{remark}
\label{remark:automataduality}
 \hfill
 \begin{enumerate}
     \item 
    Let $\aut = (Q, \Sigma, q_\init,\delta_\aut, F)$ be a deterministic reachability automaton and let $\complement{\aut}$ be the deterministic safety automaton~$(Q, \Sigma, q_\init,\delta_\aut, Q \setminus F)$. Then, $L(\complement{\aut}) = \complement{L(\aut)}$.

     \item 
    Let $\aut = (Q, \Sigma, q_\init,\delta_\aut, \col)$ be a deterministic parity automaton and let $\complement{\aut}$ be the deterministic parity automaton~$(Q, \Sigma, q_\init,\delta_\aut, q \mapsto \col(q)+1)$. Then, $L(\complement{\aut}) = \complement{L(\aut)}$.
 \end{enumerate}
\end{remark}

%%%%%%%%%%%%%%%%%%%%%%%%%%%%%%%%%%%%%%%%%%%%%%%%%%%%%%%%%%%%%%%%%%%%%%%%%%%%%%%%%%%%%%%%%%%%
%%%%%%%%%%%%%%%%%%%%%%%%%%%%%%%%%%%%%%%%%%%%%%%%%%%%%%%%%%%%%%%%%%%%%%%%%%%%%%%%%%%%%%%%%%%%
%%%%%%%%%%%%%%%%%%%%%%%%%%%%%%%%%%%%%%%%%%%%%%%%%%%%%%%%%%%%%%%%%%%%%%%%%%%%%%%%%%%%%%%%%%%%
\section{From Games under Delayed Control to Delay Games and Back}
\label{sec:transformation}

In this section, we exhibit a tight correspondence between controller in games under delayed control and Player~$I$ in delay games. 
Recall that in a game under delayed control, it is the controller whose control is delayed, i.e., she is at a disadvantage as she only gets delayed access to the action picked by environment.
In a delay game, it is Player~$I$ who is at a disadvantage as she has to grant a lookahead on her moves to Player~$O$. 
Thus, when simulating a game under delayed control by a delay game, it is natural to let Player~$I$ take the role of controller and let Player~$O$ take the role of environment. 
Also recall that the winning condition~$\win$ in a game under delayed control is formulated from controller's point-of-view: the winning condition requires her to enforce a play in $\win$. 
On the other hand, the winning condition~$L$ of a delay game is formulated from the point-of-view of Player~$O$: Player~$O$ has to enforce a play whose outcome is in $L$. 
Thus, as Player~$I$ takes the role of controller, we need to {complement} the winning condition to reflect this change in perspective: 
The set of winning outcomes for Player~$I$ in the simulating delay game is the complement of $\win$.

In the remainder of this section, we show how to simulate a game under delayed control by a delay game and then the converse, i.e., how to simulate a delay game by a game under delayed control.

\begin{transformation}
First, we transform a game under delayed control into a delay game. 
In the resulting delay game, the players simulate a play in the game under delayed control by picking actions, which uniquely induce such a play.
To formalize this, we need to introduce some notation.
Fix a game~$\arenagame = (S, s_0, S_c, S_e, \Sigma_c, \Sigma_e, \trans,\win)$.
Note that a sequence~$\sigma_0 \sigma_1 \sigma_2 \cdots \in (\Sigma_c \Sigma_e)^\omega$ induces a unique play~$\play(\sigma_0 \sigma_1 \sigma_2 \cdots ) = \pi_0 \sigma_0 \pi_1 \sigma_1 \pi_2 \sigma_2 \cdots$ in $\arenagame$ which is defined as follows:
$\pi_0 = s_0$ and $\pi_{n+1} = \trans\!\!(\pi_n, \sigma_n)$ for all $n \ge 0$.
Likewise, a finite sequence~$\sigma_0 \sigma_1 \cdots \sigma_n \in (\Sigma_c \Sigma_e)^*(\Sigma_c + \varepsilon)$ induces a unique play prefix~$\play(\sigma_0 \sigma_1 \cdots \sigma_n)$ which is defined analogously.

Now, we define the language~$L(\arenagame) \subseteq (\Sigma_c \times \Sigma_e)^\omega$ such that $ \binom{\sigma_0}{\sigma_1} \binom{\sigma_2}{\sigma_3} \binom{\sigma_4}{\sigma_5} \cdots \in L(\arenagame)$  if and only if $\play(\sigma_0 \sigma_1 \sigma_2 \cdots )$ is winning for controller.
\end{transformation}

Now, we prove the correspondence between $\arenagame$ and $\delaygame{\complement{L(\arenagame)}}$. 
The winning condition of the delay game is the complement of $L(\arenagame)$, which implements the switch of perspective described above. 

\begin{lemma}\label{lemma:fromDelayedControlToDelay}
Let $\arenagame$ be a game and $\delta \ge 0$ even. Controller wins $\arenagame$ under delay $\delta$ if and only if Player~$I$ wins $\delaygame{\complement{L(\arenagame)}}$ for $k = \frac{\delta}{2}$.
\end{lemma}

Now, we consider the converse and transform a delay game into a game under delayed control.

\begin{transformation}
Fix a delay game~$\delaygame{L}$. 
We construct a game under delayed control to simulate $\delaygame{L}$ as follows: The actions of controller are the letters in $\SigmaI$, and the actions of environment are the letters in $\SigmaO$. Thus, by picking actions, controller and environment construct the outcome of a play of $\delaygame{L}$.
As winning conditions of games under delayed control only refer to states visited by a play, but not the actions picked by the players, we reflect the action picked by a player in the state reached by picking that action.
Here, we have to require without loss of generality that $\SigmaI$ and $\SigmaO$ are disjoint.

Formally, we define~$\arenagame(L) = (S, s_0, S_c,S_e,\Sigma_c, \Sigma_e, \trans,\win)$ with $S = S_c \cup S_e$, $S_c = \set{s_0} \cup \SigmaO$, $S_e = \SigmaI$, $\Sigma_c = \SigmaI$, $\Sigma_e = \SigmaO$, $\trans\!\!(s, a) = a$ for all $s \in S_c$ and $a \in \SigmaI$, and $\trans\!\!(s,b) =  b$ for all $s \in S_e$ and $b \in \SigmaO$.
Finally, we define $\win = \set{s_0 s_1 s_2 \cdots \mid \binom{s_0}{s_1}\binom{s_2}{s_3}\binom{s_4}{s_5}\cdots \in L}$.
\end{transformation}

The following remark states that the two transformations are inverses of each other, which simplifies the proof of correctness of the second transformation.
It follows by a careful inspection of the definitions.

\begin{remark}
\label{remark:transformationsareinverses}
Let $L \subseteq (\SigmaI \times \SigmaO)^\omega$. Then, $L = {L(\arenagame(L))}$.
\end{remark}

Now, we show that the second transformation is correct, again using complementation to implement the perspective switch.

\begin{lemma}
\label{lemma:fromDelayToDelayedControl}
Let $L \subseteq (\SigmaI\times\SigmaO)^\omega$ and $k \ge 0$. Player~$I$ wins $\delaygame{L}$ if and only if controller wins $\arenagame(\complement{L})$ under delay~$2k$.
\end{lemma}

%%%%%%%%%%%%%%%%%%%%%%%%%%%%%%%%%%%%%%%%%%%%%%%%%%%%%%%%%%%%%%%%%%%%%%%%%%%%%%%%%%%%%%%%%%%%
%%%%%%%%%%%%%%%%%%%%%%%%%%%%%%%%%%%%%%%%%%%%%%%%%%%%%%%%%%%%%%%%%%%%%%%%%%%%%%%%%%%%%%%%%%%%
%%%%%%%%%%%%%%%%%%%%%%%%%%%%%%%%%%%%%%%%%%%%%%%%%%%%%%%%%%%%%%%%%%%%%%%%%%%%%%%%%%%%%%%%%%%%
\section{Results}
\label{sec:results}

\cref{lemma:fromDelayedControlToDelay} and \cref{lemma:fromDelayToDelayedControl} allow us to transfer results from delay games to games under delayed control.
Due to the definitions of strategies in games under delayed control not being dual, we consider both players independently, controller in \cref{subsec:controller} and environment in \cref{sec:environment}.

Recall that delay that allows controller to win satisfies a monotonicity property (see \cref{remark:delaycontrol}.\ref{remarkitem:delaycontrol:monotonicity}): if controller wins a game under delay~$\delta$, then also under every delay~$\delta' < \delta$.
Thus, the set of delays for which controller wins is downward-closed, i.e., it is either a finite set~$\set{0,2,4\ldots, \delta_{\max}}$ or it is equal to the set~$2\nats$ of even numbers.
In the following, we study the complexity of determining whether controller wins under all possible delays, whether she wins under a given delay, and determine bounds on~$\delta_{\max}$.

Note that winning for environment is independent of delay and boils down to the classical notion of winning delay-free games~\cite{GTW02}, which is a well-studied problem.
Hence, we disregard this problem.
However, we do discuss the relation between environment in a game under delayed control and Player~$O$ in the simulating delay game constructed in the previous section.

%%%%%%%%%%%%%%%%%%%%%%%%%%%%%%%%%%%%%%%%%%%%%%%%%%%%%%%%%%%%%%%%%%%%%%%%%%%%%%%%%%%%%%%%%%%%
%%%%%%%%%%%%%%%%%%%%%%%%%%%%%%%%%%%%%%%%%%%%%%%%%%%%%%%%%%%%%%%%%%%%%%%%%%%%%%%%%%%%%%%%%%%%
%%%%%%%%%%%%%%%%%%%%%%%%%%%%%%%%%%%%%%%%%%%%%%%%%%%%%%%%%%%%%%%%%%%%%%%%%%%%%%%%%%%%%%%%%%%%
\subsection{Controller's View}
\label{subsec:controller}

Before we present our results, we need to specify how to measure the size of games and delay games, especially how winning conditions are represented (recall that, so far, they are just $\omega$-languages).
In the following, we only consider $\omega$-regular winning conditions specified by $\omega$-automata (see \cref{subsec:automata}) or formulas of Linear Temporal Logic (LTL)~\cite{Pnueli}, which subsume the typical specification languages for winning conditions.
Hence, the size of a game~$(S, s_0, S_c, S_e, \Sigma_c, \Sigma_e, \trans, \win)$ under delayed control is given by the sum~$\size{S} + \size{\Sigma_c} + \size{\Sigma_e} + \size{\win}$, where $\size{\win}$ is the size of an automaton or LTL formula (measured in the number of distinct subformulas) representing $\win$.
Analogously, for a delay game~$\delaygame{L}$, we define the size of $L$ as the size of an automaton or LTL formula (measured in the number of distinct subformulas) representing $L$.
The bound~$k$ is encoded in binary, if necessary.

\paragraph{Safety.}

A game~$\arenagame = (S, s_0, S_c, S_e, \Sigma_c, \Sigma_e, \trans, \win)$ with winning condition~$\win$ is a safety game if $\win$ is accepted by a deterministic safety automaton. 

\begin{remark}
When Chen et al.\ introduced safety games under delayed control, they did not use automata to specify their winning plays, but instead equipped the game with a set of unsafe states and declared all those plays winning for controller that never visit an unsafe state.
It is straightforward to see that our definition is equivalent, as their definition is captured by a deterministic safety automaton with two states.
Conversely, taking the product of a game and a deterministic safety automaton yields an equivalent game with a state-based safety condition. 
\end{remark}

Our results rely on the following two bounds on the transformations presented in \cref{sec:transformation}, which are obtained by applying \cref{remark:automataduality}:
\begin{enumerate}
    \item If the winning condition~$\win$ for a game~$\arenagame$ under delayed control is given by a deterministic safety automaton with $n$ states, then the winning condition~$\complement{L(\arenagame)}$ is recognized by a deterministic reachability automaton with $n$ states.
    
    \item Dually, if the winning condition~$L \subseteq (\SigmaI \times \SigmaO)^\omega$ of a delay game is given by a deterministic reachability automaton with $n$ states, then the winning condition of the game~$\arenagame(\complement{L})$ under delayed action is recognized  by a deterministic safety automaton with $\bigo(n\cdot\size{\SigmaI})$ states.
\end{enumerate}

We begin by settling the complexity of determining whether controller wins a given safety game under every delay, which follows from the \pspace-completeness of determining whether there is a lookahead that allows Player~$O$ to win a given delay game with reachability winning condition~\cite{KleinZ14}.

\begin{theorem}
\label{thm:controlleruniversality}
The following problem is \pspace-complete: Given a safety game~$\arenagame$, does controller win $\arenagame$ under every delay~$\delta$?
\end{theorem}

Next, we give a lower bound on the complexity of determining whether controller wins a given safety game under a given delay, which is derived from a lower bound for delay games with reachability winning conditions.

\begin{theorem}
\label{thm:fixeddelta}
The following problem is \pspace-hard: Given a safety game~$\arenagame$ and $\delta$ (encoded in binary), does controller win $\arenagame$ under delay~$\delta$.
\end{theorem}

Note that we do not claim any upper bound on the problem considered in \cref{thm:fixeddelta}.
There is a trivial \twoexp upper bound obtained by hardcoding the delay into the graph of the safety game, thereby obtaining a classical delay-free safety game. 
It is open whether the complexity can be improved. 
Let us remark though that, via the correspondence to delay games presented in \cref{sec:transformation}, improvements here would also yield improvements on the analogous problem for delay games, which is open too~\cite{Zimmermann22}.

Next, we turn our attention to bounds on the delay for which controller wins.
Recall that due to monotonicity, the set of delays for which controller wins is downward-closed, i.e., it is either a finite set~$\set{0,2,4\ldots, \delta_{\max}}$ or it is equal to $2\nats$.
In the following, we present tight bounds on the value~$\delta_{\max}$.

As a consequence, we settle a conjecture by Chen et al.: They conjectured that there is some delay~$\delta_t$ (exponential in $\size{\arenagame}$), such that if controller wins $\arenagame$ under delay~$\delta_t$, then she wins under every delay.
Note that this conjecture implies that $\delta_{\max}$ is at most exponential.

The following theorem proves Chen et al.'s conjecture, while \cref{thm:delaylb} shows that $\delta_t$ must necessarily be exponential.
For $\delta_{\max}$ this means it is at most exponential for every game, and can be exponential for some games.

The following two results are again obtained from similar bounds for delay games with reachability winning conditions.

\begin{theorem}
\label{thm:universaldelta}
Let $\arenagame$ be a safety game. 
There is a $\delta_t \in \bigo(2^{\size{\arenagame}})$ such that if controller wins $\arenagame$ under delay~$\delta_t$, then she wins $\arenagame$ under every $\delta$.
\end{theorem}

Finally, we show that the exponential upper bound on $\delta_{\max}$ is tight.

\begin{theorem}
\label{thm:delaylb}
For every $n > 1$, there is a safety game~$\arenagame_n$ of size~$\bigo(n)$ such that  controller wins $\arenagame$ under delay~$2^n$, but not under delay~$2^n+2$.
\end{theorem}

%%%%%%%%%%%%%%%%%%%%%%%%%%%%%%%%%%%%%%%%%%%%%%%%%%%%%%%%%%%%%%%%%%%%%%%%%%%%%%%%%%%%%%%%%%%%
%%%%%%%%%%%%%%%%%%%%%%%%%%%%%%%%%%%%%%%%%%%%%%%%%%%%%%%%%%%%%%%%%%%%%%%%%%%%%%%%%%%%%%%%%%%%
%%%%%%%%%%%%%%%%%%%%%%%%%%%%%%%%%%%%%%%%%%%%%%%%%%%%%%%%%%%%%%%%%%%%%%%%%%%%%%%%%%%%%%%%%%%%
\paragraph{Parity.}

Next, we consider the case of $\omega$-regular winning conditions, given by deterministic parity automata.
Applying \cref{remark:automataduality} yields the following two bounds on the transformations from \cref{sec:transformation}:
\begin{enumerate}
    \item If the winning condition~$\win$ for a game~$\arenagame$ under delayed control is given by a deterministic parity automaton with $n$ states, then the winning condition~$\complement{L(\arenagame)}$ is recognized by a deterministic parity automaton with $n$ states.
    
    \item Dually, if the winning condition~$L \subseteq (\SigmaI \times \SigmaO)^\omega$ of a delay game is given by a deterministic parity automaton with $n$ states, then the winning condition of the game~$\arenagame(\complement{L})$ under delayed action is recognized  by a deterministic parity automaton with $\bigo(n\cdot\size{\SigmaI})$ states.
\end{enumerate}
Exponential lookahead is both sufficient to win all $\omega$-regular delay games that can be won and required to win some of these games~\cite{KleinZ14}.
Furthermore, determining whether there is some lookahead that allows Player~$O$ to win a given $\omega$-regular delay game is \exptime-complete~\cite{KleinZ14}.
As in the case of safety games, we can transfer these results to games under delayed control with $\omega$-regular winning conditions.

\begin{theorem}
\hfill
\begin{enumerate}
    \item The following problem is \exptime-complete: Given a game~$\arenagame$ with $\omega$-regular winning condition specified by a deterministic parity automaton, does controller win $\arenagame$ under every delay~$\delta$?
    \item Let $\arenagame$ be a game with $\omega$-regular winning condition specified by a deterministic parity automaton with $n$ states. 
There is a $\delta_t \in \bigo(2^{n^2})$ such that if controller wins $\arenagame$ under delay~$\delta_t$, then she wins $\arenagame$ under every $\delta$.
    \item For every $n > 1$, there is a game~$\arenagame_n$ of size $\bigo(n^2)$ with $\omega$-regular winning condition specified by a two-state deterministic parity automaton~$\aut_n$ such that controller wins $\arenagame$ under delay~$2^n$, but not under delay~$2^n+2$.
\end{enumerate}
\end{theorem}
 Note that the lower bound on $\delta_t$ is just a restatement of \cref{thm:delaylb}, as safety games have $\omega$-regular winning conditions.

%%%%%%%%%%%%%%%%%%%%%%%%%%%%%%%%%%%%%%%%%%%%%%%%%%%%%%%%%%%%%%%%%%%%%%%%%%%%%%%%%%%%%%%%%%%%
%%%%%%%%%%%%%%%%%%%%%%%%%%%%%%%%%%%%%%%%%%%%%%%%%%%%%%%%%%%%%%%%%%%%%%%%%%%%%%%%%%%%%%%%%%%%
%%%%%%%%%%%%%%%%%%%%%%%%%%%%%%%%%%%%%%%%%%%%%%%%%%%%%%%%%%%%%%%%%%%%%%%%%%%%%%%%%%%%%%%%%%%%
\paragraph{Linear Temporal Logic.}

Finally, one can also transfer the triply-exponential upper and lower bounds on the necessary lookahead in delay games with LTL winning conditions as well as the \threeexp-completeness of determining whether Player~$O$ wins such a delay game with respect to some lookahead~\cite{KleinZ16} to games under delayed control with LTL winning conditions.
Here, we exploit the following facts:
\begin{enumerate}
    \item If the winning condition~$\win$ for a game~$\arenagame$ under delayed control is given by an LTL formula~$\varphi$, then the winning condition~$\complement{L(\arenagame)}$ is given by an LTL formula of size~$\bigo(\size{\varphi})$.
    
    \item Dually, if the winning condition~$L \subseteq (\SigmaI \times \SigmaO)^\omega$ of a delay game is given by an LTL formula~$\varphi$, then the winning condition of the game~$\arenagame(\complement{L})$ under given action is given by an LTL formula of size $\bigo(\size{\varphi})$.
\end{enumerate}

\begin{theorem}
\hfill
\begin{enumerate}
    \item The following problem is \threeexp-complete: Given a game~$\arenagame$ with winning condition specified by an LTL formula~$\varphi$, does controller win $\arenagame$ under every delay~$\delta$?
    \item Let $\arenagame$ be a game with $\omega$-regular winning condition specified by an LTL formula~$\varphi$. 
There is a $\delta_t \in \bigo(2^{2^{2^{\size{\varphi} + \size{\arenagame}}}})$ such that if controller wins $\arenagame$ under delay~$\delta_t$, then she wins $\arenagame$ under every $\delta$.
    \item For every $n > 1$, there is a game~$\arenagame_n$ of size $\bigo(n^2)$ with winning condition specified by an LTL formula~$\varphi_n$ of size~$\bigo(n^2)$ such that controller wins $\arenagame$ under delay~$2^{2^{2^n}}$, but not under delay~$2^{2^{2^n}}+2$.
\end{enumerate}
\end{theorem}

To conclude, let us just remark that the results presented here also allow us to transfer results obtained for delay games with quantitative winning conditions~\cite{KleinZ16,Zimmermann16,Zimmermann17} to games under delayed control with quantitative winning conditions.
In fact, our result works for any winning condition, as long as the two transformations described in \cref{sec:transformation} are effective.

%%%%%%%%%%%%%%%%%%%%%%%%%%%%%%%%%%%%%%%%%%%%%%%%%%%%%%%%%%%%%%%%%%%%%%%%%%%%%%%%%%%%%%%%%%%%
%%%%%%%%%%%%%%%%%%%%%%%%%%%%%%%%%%%%%%%%%%%%%%%%%%%%%%%%%%%%%%%%%%%%%%%%%%%%%%%%%%%%%%%%%%%%
%%%%%%%%%%%%%%%%%%%%%%%%%%%%%%%%%%%%%%%%%%%%%%%%%%%%%%%%%%%%%%%%%%%%%%%%%%%%%%%%%%%%%%%%%%%%
\subsection{Environment's View}
\label{sec:environment}

In \cref{sec:transformation}, we proved a tight correspondence between controller in a game under delayed control and Player~$I$ in a delay game.
Thus, it is natural to ask whether environment and Player~$O$ also share such a tight correspondence.
A first indication that this is not the case can be obtained by considering the determinacy of these games:
While delay games with Borel winning conditions are determined~\cite{KleinZ15}, even safety games under delayed action are not necessarily determined~\cite{ChenFLMZ21}.

Upon closer inspection, this is not surprising, as the strategies in games under delayed control are not dual between the players:
controller is at a disadvantage as she only gets delayed access to the actions picked by environment while environment does not benefit from this disadvantage. 
He does not get access to the actions picked by controller in advance.
In a delay game however, the strategy definitions are completely dual: Player~$I$ has to grant lookahead on her moves which Player~$O$ gets access to.
Thus, environment is in a weaker position than Player~$O$.\footnote{The difference can be formalized in terms of the information the players have access to: safety games under delay are incomplete-information games while delay games are complete-information games. Although interesting, we do not pursue this angle any further.}

In this section, we study the correspondence between environment and Player~$O$ in detail by formally proving that environment is weaker than Player~$O$.

\begin{lemma}
\label{lemma:envtoinput}
Let $\arenagame$ be a safety game. If environment wins $\arenagame$ then Player~$O$ wins $\delaygame{\complement{L(\arenagame)}}$ for every $k$.
\end{lemma}

Now, we show that the converse direction fails.

\begin{lemma}
\label{lemma:envfailure}
There is a safety game~$\arenagame$ such that Player~$O$ wins $\delaygame{\complement{L(\arenagame)}}$ for some $k$, but environment does not win $\arenagame$.
\end{lemma}

\begin{proof}
Let $\arenagame$ be the safety game depicted in \cref{fig:env}. 
With each move, the players place a coin (by either picking heads or tails) and environment wins a play by correctly predicting the second action of controller with his first action.
Clearly, environment has no winning strategy in $\arenagame$ because he has no access to future moves of controller.
Stated differently, if environment picks $h$ ($t$) in his first move, then the play in which the second action of controller is $t$ ($h$) is winning for controller.\footnote{Note that under any delay~$\delta > 0$, controller cannot do this strategically, as she has to fix her first two actions in advance. But, as environment has no access to these fixed actions, he cannot react to them strategically.}

\begin{figure}
    \centering
\begin{tikzpicture}[ultra thick]
    
    \draw[gray!20, rounded corners,line width=2mm] (-1.15,-1) rectangle (13.15,1.3);  
    
    \draw[fill=gray!20,gray!20] (1.5,-1) rectangle (4.5,1.3); 
    \draw[fill=gray!20,gray!20] (7.5,-1) rectangle (10.5,1.3); 
    
    \node[black] at (0,1) {\small \bf C};
    \node[black] at (3,1) {\small \bf E};
    \node[black] at (6,1) {\small \bf C};
    \node[black] at (9,1) {\small \bf E};
    \node[black] at (12,1) {\small \bf C};
    
    \node[mystate] (1) at (0,0) {\phantom{v}};
    \node[mystate] (2) at (3,0) {\phantom{v}};
    \node[mystate] (3) at (6,-.5) {\phantom{v}};
    \node[mystate] (4) at (6, .5) {\phantom{v}};
    
    \node[mystate,fill = black] (5) at (9,-.5) {\phantom{v}};
    \node[mystate] (6) at (9, .5) {\phantom{v}};
    
    \node[mystate,fill = black] (7) at (12,-.5) {\phantom{v}};
    \node[mystate] (8) at (12, .5) {\phantom{v}};
    
    \path[-stealth]
    (-.75,0) edge (1)
    (1) edge[near start] node[above] {$h,t$} (2)
    
    (2) edge[near start] node[above] {$h$} (4)
    (2) edge[near start] node[below] {$t$} (3)
    
    (3) edge[near start] node[below] {$h$} (5)  
    (3.north east) edge[near start] node[below] {$t$} (6)  

    (4) edge[near start] node[above] {$h$} (6)  
    (4.south east) edge[near start] node[above] {$t$} (5)  

    (5) edge[bend left=20, near start] node[below] {$h,t$} (7)  
    (7) edge[bend left=20, near start] node[above] {$h,t$} (5)  

    (6) edge[bend left=20, near start] node[below] {$h,t$} (8)  
    (8) edge[bend left=20, near start] node[above] {$h,t$} (6)
    ;

\end{tikzpicture}
    \caption{A safety game that environment does not win, but Player~$O$ wins the associated delay game. The initial state is marked by an arrow and the unsafe vertices are black. Note that both players have the actions~$h$ and $t$ available.}
    \label{fig:env}
\end{figure}
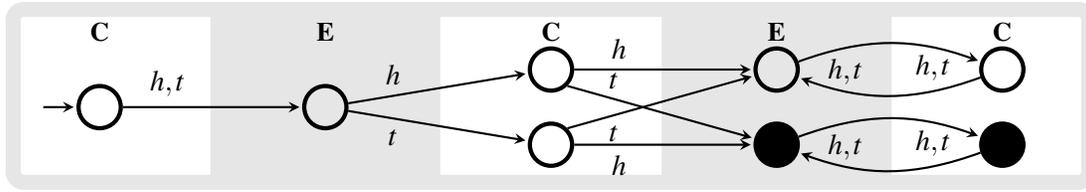

Now, we consider the delay game~$\delaygame{\complement{L(\arenagame)}}$ for $k = 1$.
Recall that the winning condition~$\complement{L(\arenagame)}$ contains the winning plays for Player~$O$, i.e., we have $\binom{ \sigma_0 \sigma_2 \sigma_4 \cdots }{ \sigma_1 \sigma_3 \sigma_5 \cdots} \in \complement{L(\arenagame)}$ if and only if $\sigma_{1} \neq \sigma_{2}$.
It is easy to see that Player~$O$ has a winning strategy in $\delaygame{\complement{L(\arenagame)}}$ by simply flipping the second letter picked by Player~$I$.
This is possible since Player~$I$ has to provide two letters during the first round.
\end{proof}

\begin{remark}
\label{remark:undetermined}
The safety game~$\arenagame$ depicted in \cref{fig:env} is in fact undetermined under every delay~$\delta > 0$.
In the proof of \cref{lemma:envfailure}, we have already established that environment does not win $\arenagame$.
Now, under every delay~$\delta > 0$, controller has to fix at least two actions before getting access to the first action picked by environment. 
This implies that there is, for every strategy for controller under delay~$\delta$, at least one consistent play that is losing for her, i.e., a play in which environment picks $h$ ($t$) if the second move fixed by controller is $t$ ($h$).
Thus, no strategy is winning for controller under delay $\delta$.

Let us remark that, according to our definition of environment strategies, he is not able to enforce a losing play for controller (the game is  undetermined after all), as he does not get access to the second action fixed by controller.
Also, this is again the difference to delay games: Player~$O$ has access to these first two actions when making his first move, and is thereby able to win.
\end{remark}

The full relation between games under delayed control and delay games is depicted in \cref{fig:transf}, restricted to Borel winning conditions (note that both transformations described in \cref{sec:transformation} preserve Borelness).
The equivalence between controller winning the game under delayed control and Player~$I$ winning the corresponding delay game has been shown in \cref{lemma:fromDelayedControlToDelay} and \cref{lemma:fromDelayToDelayedControl}.
Also, \cref{lemma:fromDelayToDelayedControl} and \cref{remark:transformationsareinverses} imply that undetermined safety games under delayed control and those won by environment get transformed into delay games that are won by Player~$O$.
Finally, \cref{lemma:fromDelayedControlToDelay} and \cref{remark:transformationsareinverses} imply that delay games won by Player~$O$ get transformed into undetermined safety games under delayed control or to ones that are won by environment.

\begin{figure}
    \centering
    \begin{tikzpicture}[ultra thick]

        \begin{scope}
            \draw[fill=gray!10, blur shadow={shadow blur steps=10}](0,-.55)circle[x radius=5cm, y radius=.6cm];
            \clip(0,-.55)circle[x radius=5.1cm, y radius=.619cm];
            \draw[] (-1.6,-1.55) -- (-1.6,.75);
            \draw[] (1.6,.75) -- (1.6,-1.55);
        \end{scope}

        \begin{scope}
                
            \draw[fill=gray!10, blur shadow={shadow blur steps=10}](0,-2.5)circle[x radius=5cm, y radius=.6cm];
            \clip(0,-2.5)circle[x radius=5.1cm, y radius=.619cm];    
            \draw[] (0,-1.5) -- (0,-4.5);
        \end{scope}
        
        \node[] at (0,-.55) {\large undetermined};
        \node[] at (-3,-.55) {\large C wins};
        \node[] at ( 3,-.55) {\large E wins};
        \node[] at (-2, -2.5) {\large $I$ wins};
        \node[] at ( 2, -2.5) {\large $O$ wins};

        \path[-stealth]
        % c to I
        (-3.25,-.75) edge[line width= 1mm,bend right,gray] (-3.25,-2.35)
        % I to c
        (-2.75,-2.25) edge[line width= 1mm,bend right,gray] (-2.75,-.75)
        %e to o
        (3.25,-.75) edge[line width= 1mm,gray] (2.2,-2.3)
        %un to o
        (0,-.95) edge[gray, line width=1mm] (1,-2.3)
        % fork
        (1.6,-2.2) edge[line width= 1mm,gray,-] (1.6,-1.7) 
        (1.6,-1.8) edge[line width= 1mm,gray] (2.2,-.75)
        (1.6,-1.8) edge[line width= 1mm,gray] (1,-.75)
        ;
        
        \node[align=left,anchor=west] at (-8.5,-.4) {\large Games under\\\large  delayed control};
        \node[align=left,anchor=west] at (-8.5,-3) {\large Delay games};
        
    \end{tikzpicture}
    \caption{The relation between games under delayed control and delay games with Borel winning conditions. The upper ellipsis contains pairs~$(\arenagame, \delta)$ consisting of a game~$\arenagame$ under delayed control and a fixed delay~$\delta$; the lower one contains delay games~$\delaygame{L}$ for some fixed~$k$. 
    The arrows represent the two transformations described in \cref{sec:transformation}.}
    \label{fig:transf}
\end{figure}
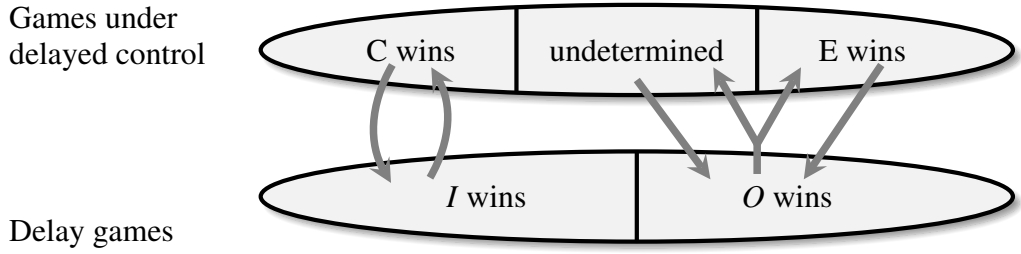

%%%%%%%%%%%%%%%%%%%%%%%%%%%%%%%%%%%%%%%%%%%%%%%%%%%%%%%%%%%%%%%%%%%%%%%%%%%%%%%%%%%%%%%%%%%%
%%%%%%%%%%%%%%%%%%%%%%%%%%%%%%%%%%%%%%%%%%%%%%%%%%%%%%%%%%%%%%%%%%%%%%%%%%%%%%%%%%%%%%%%%%%%
%%%%%%%%%%%%%%%%%%%%%%%%%%%%%%%%%%%%%%%%%%%%%%%%%%%%%%%%%%%%%%%%%%%%%%%%%%%%%%%%%%%%%%%%%%%%
\section{Refining the Correspondence: Sure Winning and Almost Sure Winning}
\label{Sec:AlmostSure}

It should be noted that the above transformations of games under delayed control into delay games and vice versa hinge on the fact that environment in the game under delayed control could, though lacking recent state information to do so strategically, by mere chance play the very same actions that the informed Player~$O$ in the delay game plays in his optimal adversarial strategy. 
That this constitutes a fundamental difference becomes apparent if we consider almost sure winning instead of sure winning. Almost sure winning calls for the existence of a mixed strategy that wins with probability~$1$, i.e., may fail on a set of plays with measure~$0$. This is different from sure winning in the sense of the definition of winning strategies for games under delayed control in \cref{subsec:games}, which calls for a strategy that never fails.

\begin{remark}
We introduce mixed strategies for games under delayed control only, as delay games (with Borel winning conditions) are determined, which means that mixed strategies do not offer any advantage over pure strategies as introduced in \cref{subsec:delaygames}.
\end{remark}

Given an even $\delta \ge 0$, a mixed strategy for controller in $\arenagame$ under delay~$\delta$ is a pair $(\alpha, \stratC)$ where $\alpha \in \prob{(\Sigma_c)^{\frac{\delta}{2}}}$ is a probability distribution over $(\Sigma_c)^{\frac{\delta}{2}}$ and $\stratC \colon \prefsC \rightarrow \prob{\Sigma_c}$ maps play prefixes ending in $S_c$ to probability distributions over actions of controller.
A mixed strategy for environment is a mapping~$\stratE\colon \prefsE \rightarrow \prob{\Sigma_e}$.

The notion of consistency of a play with a strategy simply carries over, now inducing a Markov chain due to the probabilistic nature of the strategies. We say that a mixed strategy for controller (environment) \emph{wins almost surely} if and only if it wins against any strategy of its opponent environment (controller) with probability~$1$, i.e., if and only if the winning condition is satisfied with probability~$1$ over the Markov chain induced by the game and the particular strategy combination. In this section, we write sure winning for winning as defined in \cref{sec:prels}, as is usual for games with randomized strategies.

The notion of almost sure winning alters chances for the players substantially by excluding the possibility of reliably playing an optimal strategy though lacking the information for doing so due to delayed observation. This can be seen from the following lemma, stating a fundamental difference between controller's power in games under delayed control and Player~$I$'s power in the corresponding delay games.

\begin{lemma}\label{lemma:almost-sure-win-vs-sure-loss}
    There is a game~$\arenagame$ under delayed control such that controller wins $\arenagame$ almost surely under some delay~$\delta$ while Player~$O$ (not Player~$I$, which is the player corresponding to controller) wins the corresponding delay game~$\delaygame{\complement{L(\arenagame)}}$ for $k = \frac{\delta}{2}$, and surely so. 
\end{lemma}
 
\begin{proof}
    Consider the reachability game in \cref{fig:C-almost-surely-vs-O-surely} under delay~$2$ (or any larger delay). 
    Intuitively, the players place a coin in each round (by picking either heads to tails with each move) and controller wins a play if the black state is visited, which happens if she selects a different coin placement than chosen by environment in the previous move.

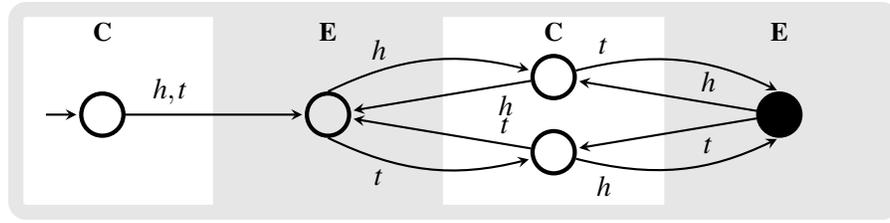
\begin{figure}
    \centering
\begin{tikzpicture}[ultra thick]
    
     \draw[gray!20, rounded corners,line width=2mm] (-1.15,-1.3) rectangle (10.5,1.4);  
    
    \draw[fill=gray!20,gray!20] (1.5,-1.3) rectangle (4.5,1.4); 
    \draw[fill=gray!20,gray!20] (7.5,-1.3) rectangle (10.5,1.4); 
    
    \node[black] at (0,1.1) {\small \bf C};
    \node[black] at (3,1.1) {\small \bf E};
    \node[black] at (6,1.1) {\small \bf C};
    \node[black] at (9,1.1) {\small \bf E};

     \node[mystate] (start) at (-0,0) {\phantom{v}};
     \node[mystate] (et) at (3,0) {\phantom{v}};
     \node[mystate,fill] (eb) at (9,0) {\phantom{v}};
     \node[mystate] (ct) at (6,.5) {\phantom{v}};
     \node[mystate] (cb) at (6,-.5) {\phantom{v}};
    
    \path[-stealth]
    (-.75,0) edge (start)

     (start) edge node[above,near start] {$h,t$} (et)  
     (et.north) edge[bend left=20,near start] node[above] {$h$} (ct)  
     (et.south) edge[bend right=20,near start] node[below] {$t$} (cb)  
     (eb) edge[bend left=0,near start] node[above] {$h$} (ct)  
     (eb) edge[bend left=0,near start] node[below] {$t$} (cb)  

     (ct) edge[bend left=0,very near start] node[below] {$h$} (et)  
     (ct) edge[bend left=20, very near start] node[above] {$t$} (eb.north)  
     (cb) edge[bend right=0, very near start] node[above] {$t$} (et)  
     (cb) edge[bend right=20, very near start] node[below] {$h$} (eb.south);

\end{tikzpicture}
    \caption{A reachability game that, under any positive delay, is won by controller almost surely via the simple randomized strategy of coin tossing (thus randomly generating head and tail events $h$ and $t$), but won by player $O$ surely if interpreted as a delay game due to the lookahead on Player~$I$'s actions granted to Player~$O$. The initial state is marked by an arrow and controller wins if and only if the black vertex is visited at least once.}
    \label{fig:C-almost-surely-vs-O-surely}
\end{figure}

    Under any even (by definition) positive delay, controller wins this game with probability~$1$, i.e., almost surely, by a simple randomized strategy of coin tossing: by in each step randomly selecting action~$h$ or $t$ with positive probability each, an eventual visit of the black state is guaranteed with probability~$1$, irrespective of being uninformed about environment's preceding move due to the delay.

    The corresponding delay game~$\delaygame{\complement{L(\arenagame)}}$ for $k = \frac{\delta}{2}$, however, is easily won by Player~$O$, because in delay games, the delayed Player~$I$ grants a lookahead to Player~$O$. Hence, Player~$O$ can, due to the delay, already see the next move of Player~$I$ such that he can simply copy the next coin placement by Player~$I$, safely staying in the non-black states and thereby win. 
\end{proof}

Note that \cref{lemma:almost-sure-win-vs-sure-loss} implies that the previously observed correspondence between Player~$I$ and controller breaks down when considering almost sure winning strategies instead of just sure winning strategies: 
Games under delayed control for which Player~$O$ wins the corresponding delay game, are no longer either undetermined or won by environment, but may well be won by controller almost surely.

This consequently refines the correspondence between games under delayed control and delay games shown in \cref{fig:transf} as follows.

\begin{theorem}
\label{theorem:refined-correspondence}
    Given a game~$\arenagame$ and an even $\delta \ge 2$,  the following correspondences between $\arenagame$ and the corresponding delay game~$\delaygame{\complement{L(\arenagame)}}$ for $k = \frac{\delta}{2}$ hold:
    \begin{enumerate}
        \item\label{item:csure} Controller surely wins $\arenagame$ under delay~$\delta$ if and only if Player~$I$ surely wins $\delaygame{\complement{L(\arenagame)}}$.

        \item\label{item:calmostsure} If controller almost surely wins $\arenagame$ under delay $\delta$ but cannot surely win $\arenagame$ under delay $\delta$ then Player~$O$ surely wins $\delaygame{\complement{L(\arenagame)}}$.
        
        \item\label{item:e} If environment surely or almost surely wins $\arenagame$ under delay $\delta$ then Player~$O$ wins $\delaygame{\complement{L(\arenagame)}}$.
        
        \item\label{item:undet} If $\arenagame$ is undetermined under delay $\delta$ with respect to  almost sure winning strategies then Player~$O$ wins $\delaygame{\complement{L(\arenagame)}}$.
        
        \item\label{item:nonempty} All the aforementioned classes are non-empty, i.e., there exist games under delayed control where controller wins, where controller wins almost surely (but not surely), where environment wins surely, where environment wins almost surely (but not surely), and games which are undetermined with respect to  almost-sure winning strategies.
    \end{enumerate}
    The above correspondences are depicted in \cref{fig:refined-correspondence}.
\end{theorem}

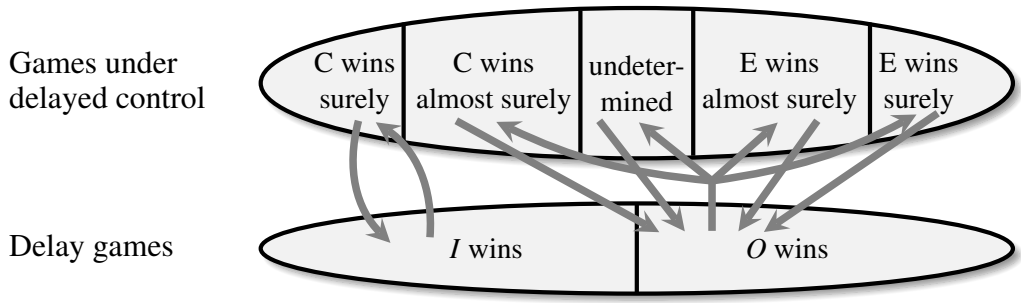
\begin{figure}[t]
    \centering
    \begin{tikzpicture}[ultra thick]

        \begin{scope}
            
            \draw[fill=gray!10, blur shadow={shadow blur steps=10}](0,0)circle[x radius=5cm, y radius=1cm];
            \clip(0,0)circle[x radius=5.1cm, y radius=1.019cm];
             \draw[] (-.75,-1) -- (-.75,1);
             \draw[] (.75,1) -- (.75,-1);
             \draw[] (-3.1,-1) -- (-3.1,1);
             \draw[] (3.1,1) -- (3.1,-1);
        \end{scope}

        \begin{scope}
           
            \draw[fill=gray!10, blur shadow={shadow blur steps=10}](0,-2.2)circle[x radius=5cm, y radius=.6cm];
            \clip(0,-2.2)circle[x radius=5.1cm, y radius=.619cm];        
            \draw[] (0,-1) -- (0,-4.5);
        \end{scope}
        
        \node[align = center] at (0,0) {undeter-\\ mined};
        \node[align = center] at (-3.75,0) {C wins\\ surely};
        \node[align = center] at ( 3.75,0) {E wins\\ surely};
        \node[align = center] at (-1.9,0) {C wins\\ almost surely};
        \node[align = center] at ( 1.9,0) {E wins\\ almost surely};
        \node[] at (-2, -2.2) {$I$ wins};
        \node[] at ( 2, -2.2) {$O$ wins};

        \path[-stealth]
        % c to I
        (-3.7,-.5) edge[line width= 1mm,bend right,gray] (-3.25,-2.15)
        % I to c
        (-2.75,-2.05) edge[line width= 1mm,bend right,gray] (-3.5,-.5)
        
        (4,-.4) edge[line width= 1mm,gray] (1.66,-2)
        (2.4,-.5) edge[line width= 1mm,gray] (1.33,-2)
        (-0.5,-.5) edge[gray, line width=1mm] (.66,-2)
        (-2.4,-.5) edge[gray, line width=1mm] (.33,-2)

        % fork
        (1,-1.3) edge[line width= 1mm,gray,-] (1,-2)
        (1,-1.3) edge[bend left = 10, line width= 1mm,gray] (-1.9,-.5)
        (1,-1.3) edge[line width= 1mm,gray] (0,-.5)
        (1,-1.3) edge[line width= 1mm,gray] (1.9,-.5)
        (1,-1.3) edge[bend right = 10, line width= 1mm,gray] (3.75,-.4)
        ;
        
        \node[align=left,anchor=west] at (-8.5,0) {\large Games under\\\large delayed control};
        \node[align=left,anchor=west] at (-8.5,-2.2) {\large Delay games};
        
    \end{tikzpicture}
    \caption{The relation between safety games under delayed control and delay games with Borel winning conditions. The upper ellipsis contains pairs~$(\arenagame, \delta)$ consisting of a game~$\arenagame$ under delayed control and a fixed delay~$\delta$; the lower one contains delay games~$\delaygame{L}$ for some fixed~$k$. 
    The arrows represent the two transformations described in \cref{sec:transformation}.}
    \label{fig:refined-correspondence}
\end{figure}

Item~\ref{item:calmostsure}.\ of the above lemma is of particular interest, as it expresses a delay-related strengthening of controller relative to Player~$I$, letting controller win almost surely where Player~$I$ looses for sure.
The correspondence between controller and Player~$I$ observed in the deterministic setting thus breaks down when almost sure winning is considered and mixed strategies are permitted.
\begin{remark}
    In contrast to games under delayed control, where mixed strategies provide additional power to both the controller and the environment, the notions of sure winning and almost sure winning coincide for delay games (with Borel winning conditions) due to their determinacy~\cite{KleinZ15}. Admitting mixed strategies (and almost sure winning) does not provide additional power to either of the two players in a delay game, as the determinacy result always implies existence of an infallible pure strategy for one of the players.
\end{remark}

%%%%%%%%%%%%%%%%%%%%%%%%%%%%%%%%%%%%%%%%%%%%%%%%%%%%%%%%%%%%%%%%%%%%%%%%%%%%%%%%%%%%%%%%%%%%
%%%%%%%%%%%%%%%%%%%%%%%%%%%%%%%%%%%%%%%%%%%%%%%%%%%%%%%%%%%%%%%%%%%%%%%%%%%%%%%%%%%%%%%%%%%%
%%%%%%%%%%%%%%%%%%%%%%%%%%%%%%%%%%%%%%%%%%%%%%%%%%%%%%%%%%%%%%%%%%%%%%%%%%%%%%%%%%%%%%%%%%%%
\section{Conclusion}
\label{sec:conc}

We have compared delay games~\cite{KleinZ14} and games under delayed control~\cite{ChenFLMZ21}, two types of infinite games aiming to model asynchronicity in reactive synthesis,  and have exhibited the differences in definitions and charted the relation between them with respect to  both deterministic and randomized strategies: 
One can efficiently transform a game under delayed control into a delay game such that controller wins the game under delayed control with delay~$\delta$ by a deterministic strategy if and only if Player~$I$ wins the resulting delay game with lookahead of size~$\frac{\delta}{2}$. 
Dually, one can efficiently transform a delay game into a game under delayed control such that Player~$I$ wins the delay game with lookahead of size~$\delta$ if and only if controller wins the resulting game under delayed control with delay~$2\delta$ by a deterministic strategy. 
These results allow us to transfer known complexity results and bounds on the amount of delay from delay games to games under delayed control, for which no such results were known, when considering deterministic strategies.
We also proved that the analogous results fail in the setting of randomized strategies and almost sure winning conditions, as well as for the relation between environment and Player~$O$, both under deterministic and randomized strategies.

\textbf{Acknowledgements:}
Martin Fränzle has been supported by Deutsche Forschungsgemeinschaft under grant no.\ DFG FR 2715/5-1 ``Konfliktresolution und kausale Inferenz mittels integrierter sozio-technischer Modellbildung''. Sarah Winter is a postdoctoral researcher at F.R.S.-FNRS.  
Martin Zimmermann has been supported by DIREC – Digital Research Centre Denmark.

%%%%%%%%%%%%%%%%%%%%%%%%%%%%%%%%%%%%%%%%%%%%%%%
%%%%%%%%%%%%%%%%%%%%%%%%%%%%%%%%%%%%%%%%%%%%%%%
%%%% SWITCH 3
\bibliographystyle{eptcs}
%%%%
%\bibliographystyle{plain}
%%%% END SWITCH 3

\bibliography{bib}

\end{document}